\documentclass[conference]{IEEEtran}
\usepackage{amsthm}
\usepackage{multirow}
\usepackage{balance}
\usepackage{color}
\usepackage{cite}
\ifCLASSINFOpdf
   \usepackage[pdftex]{graphicx}
\graphicspath{ {figures/} }
\else
\fi
\usepackage[font=footnotesize,caption=false,subrefformat=parens,labelformat=parens]{subfig}
\usepackage{stfloats}
\usepackage{amsmath}
\usepackage{multirow}
\usepackage{pifont}
\newcommand{\cmark}{\ding{51}}%
\newcommand{\xmark}{\ding{55}}%
\usepackage{xcolor}
\usepackage{algorithmic}
\usepackage{algorithm}
\usepackage{algorithmic}
\usepackage{cite}
\usepackage{times,amsmath}
\DeclareMathOperator*{\argmax}{arg\,max}

\usepackage{color}
\usepackage{algorithm}
\usepackage{algorithmic}
\usepackage{url}
\usepackage{wrapfig,lipsum,booktabs}
\usepackage{lipsum}
\usepackage{algorithmic}
\usepackage{graphicx}
\usepackage{array}
\usepackage{url}
\newtheorem{prop}{Proposition}
\newtheorem{cor}{Corollary}
\renewcommand{\baselinestretch}{0.905}
	\newcommand\blfootnote[1]{%
		\begingroup
		\renewcommand\thefootnote{}\footnote{#1}%
		\addtocounter{footnote}{-1}%
		\endgroup
	}
\begin{document}
\title{RSS-Based Q-Learning for Indoor UAV Navigation
}
\author{\IEEEauthorblockN{Md Moin Uddin Chowdhury, Fatih Erden, and Ismail Guvenc}
Dept. Electrical and Computer Engineering, NC  State University, Raleigh, NC, 27606\\
Email: \{mchowdh, ferden, iguvenc\}@ncsu.edu
}


\maketitle
\begin{abstract}
In this paper, we focus on the potential use of unmanned aerial vehicles (UAVs) for search and rescue (SAR) missions in GPS-denied indoor environments. We consider the problem of navigating a UAV to a wireless signal source, e.g., a smartphone or watch owned by a victim. We assume that the source periodically transmits RF signals to nearby wireless access points. Received signal strength (RSS) at the UAV, which is a function of the UAV and source positions, is fed to a Q-learning algorithm and the UAV is navigated to the vicinity of the source. Unlike the traditional location-based Q-learning approach that uses the GPS coordinates of the agent, our method uses the RSS to define the states and rewards of the algorithm. It does not require any a priori information about the environment. These, in turn, make it possible to use the UAVs in indoor SAR operations. Two indoor scenarios with different dimensions are created using a ray tracing software. Then, the corresponding heat maps that show the RSS at each possible UAV location are extracted for more realistic analysis. Performance of the RSS-based Q-learning algorithm is compared with the baseline (location-based) Q-learning algorithm in terms of convergence speed, average number of steps per episode, and the total length of the final trajectory. Our results show that the RSS-based Q-learning provides competitive performance with the location-based Q-learning.
\looseness=-1
\end{abstract}
\section{Introduction}
\label{sec:Intro}
\blfootnote{This work is supported by NSF under the award CNS-1453678. The authors would like to thank Bekir S. Ciftler and Adem Tuncer for their initial inputs.}
Thanks to the extensive studies and massive cost reduction in manufacturing, the interest in the use of unmanned aerial vehicles (UAVs) is expected to increase significantly in the upcoming years. Besides their widespread recreational and military use, UAVs have already started to show up in civilian applications including but not limited to precision agriculture, infrastructure health monitoring, packet delivery, restoring service after natural disasters, patrolling missions, and search and rescue (SAR) operations~\cite{Martins2019,acuna}.

Deployment of UAVs can make a big difference in SAR missions by providing information and data about the environment or an injured or lost person, improving network access, delivering first aid equipment, among others. UAVs can be utilized by emergency services or rescue teams in the aftermath of a disaster (e.g., a hurricane or earthquake) and can help the first responders make better decisions and save time. However, due to the unavailability of a suitable data link or precise maneuver requirements that are sometimes outside human capabilities, human control over the UAVs may not be possible~\cite{Becerra}. Thus, it is critical to develop effective technologies and algorithms to enable the UAVs to perform complicated tasks autonomously. One issue with the autonomous use of UAVs in SAR missions is that, most of the time, the prior knowledge regarding the environment is limited, if not completely unavailable. Moreover, the environment may change with time or the models defining the target and its location may not be accurate or descriptive enough. Therefore, a UAV is required to interact with the environment, learn and make decisions by itself. Reinforcement learning (RL), which is a class of machine learning (ML) algorithms, may help to overcome these issues.
\begin{figure}[t]
\centering{\includegraphics[width=0.92\columnwidth]{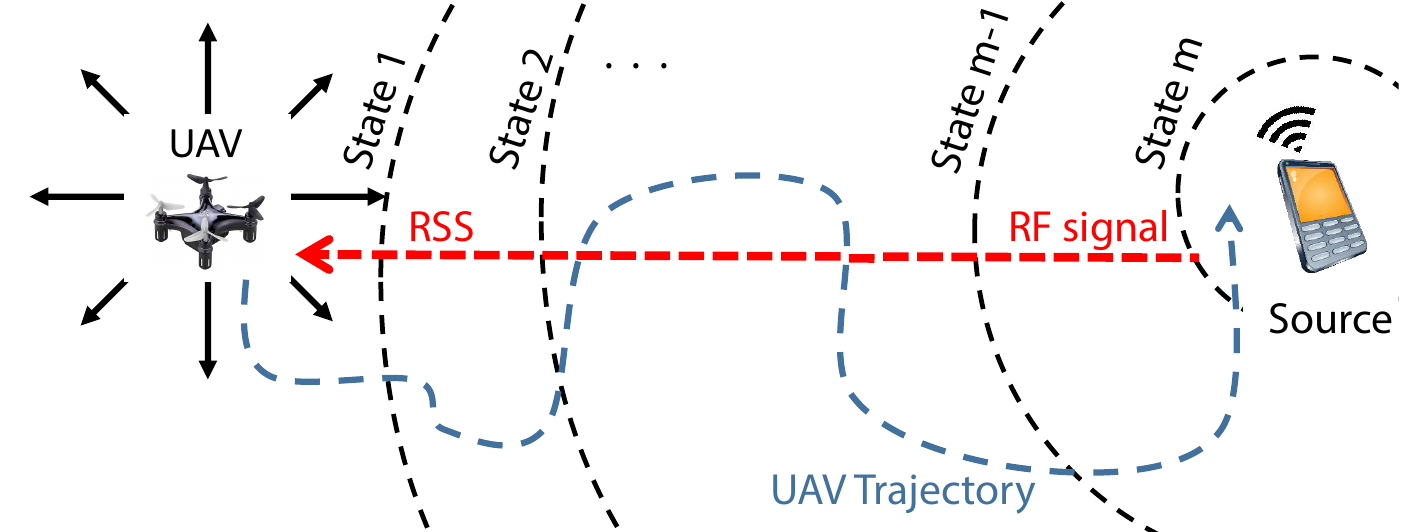}}
    \caption{Navigation of a UAV to a wireless source.}
    \label{fig:System}
\end{figure}
In RL, an agent learns in an interactive environment by using feedback from its actions and experiences. Usually, the environment is modeled as a Markov decision process (MDP) to leverage the dynamic programming technique that is used by the RL algorithms. The studies that do not make use of the ML either use exact models of the environment or assume the accurate information of the environment is predictable\cite{gesbert2018}. On the contrary, a branch of RL, known as Q-learning, requires a little or no prior/explicit knowledge of the environment. Q-learning is an off-policy RL algorithm which aims to find the best action to take given the current state. It learns from actions that are not known to the current policy by taking random actions and seeks to learn a policy that maximizes the total reward.

RL algorithms have already been widely studied in UAV-related researches as in many other fields of robotics. In~\cite{Iman2016}, a model-based RL algorithm, TEXPLORE, is used for the autonomous navigation of UAVs. The value function is updated from a model of the environment, while also taking battery life into consideration. It is shown that their method learns faster than the traditional table-based Q-learning due to its parallel architecture. Pham \textit{et al}.~\cite{Pham2018} use Q-learning to navigate the UAVs by defining states based on the UAV location. It is assumed that the UAV can observe its state at any position.
In~\cite{Wang2017}, GPS signal and sensory information of the local environment are used in deep RL for UAV navigation tasks in outdoor environments.  In~\cite{Rodriguez2019}, deep Q-learning is used for the autonomous landing of UAVs on a moving platform. In \cite{acuna}, RF signals from devices are used to estimate users location using random-forest based ML technique. 

There are recent promising attempts of navigating UAVs in GPS-denied indoor environments using image processing based techniques. In~\cite{KimC15a}, images from a single camera are input to a convolutional neural network (ConvNet) to learn a control strategy to find a specific target. In~\cite{Gandhi2017}, monocular images are used in a deep neural network to navigate a UAV while avoiding crashes. Negative flying data created from real collisions are used during training along with the positive data, and  all training is done offline. In~\cite{DBLP:journals/corr/SadeghiL16}, RGB images are fed to a deep ConvNet based learning method to enable UAVs to have collision-free indoor flights, again with offline training. 


\begin{figure}[t]
\vspace{-4mm}
\centering{\includegraphics[width=0.67\linewidth]{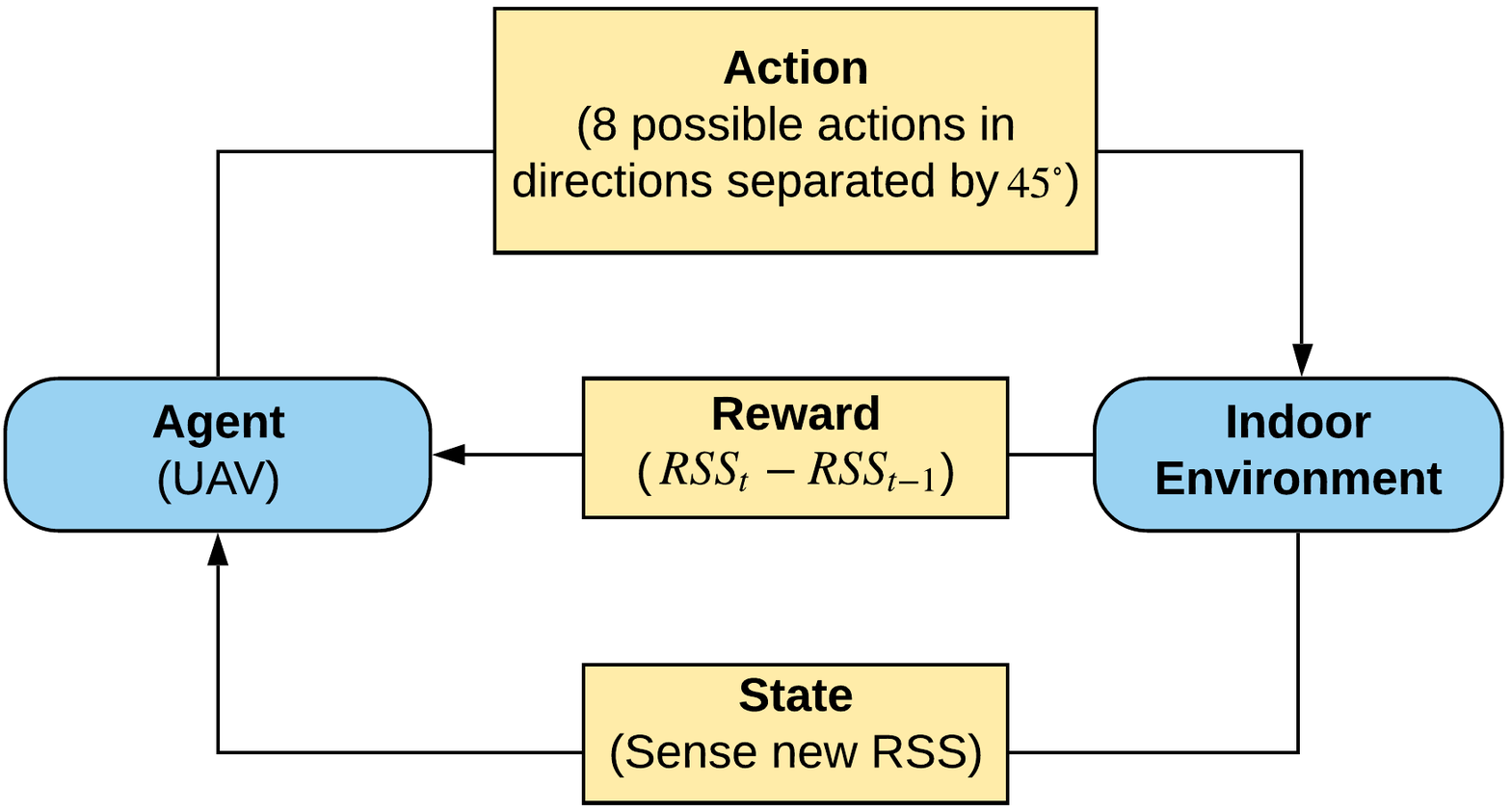}}\vspace{-5mm}
    \caption{Agent-environment interaction in the (RSS-based) Q-learning.}
    \label{RL_algo}
    \vspace{-2mm}
\end{figure}

Motivated by the above discussion, in this paper, we propose a new method for autonomous navigation of UAVs indoors using Q-learning. Smart devices (e.g., a smartphone) can be used to locate a victim in a SAR scenario through the propagated RF signals~\cite{iotpublicsafety}. Presently,  smart devices can continuously transmit RF signals to discover nearby APs. Furthermore, a smart device can be forced to transmit wireless signals in case of emergency~\cite{wang,Pu2019}. Based on this fact, unlike the location-based Q-learning, our approach uses RSS values instead of UAV location information while deciding future actions to navigate the UAV towards the target. It does not require any prior knowledge of the environment. There is also no need for an exact mathematical representation of the target or mapping of the environment to locate the target. 

A high-level view of the system architecture is shown in Fig.~\ref{fig:System}. The receiver mounted on a UAV continuously senses the environment and picks up the RF signals from a remote wireless transmitter referred to as the source. A unique state label is assigned to the RSS value at the current position. Rewards in the Q-learning algorithm are also defined as a function of successive RSS values sensed in the current and previous positions, and Q-table is updated accordingly. Finally, the UAV takes one of the possible eight actions in different directions separated by $45^{\circ}$. The proposed RSS-based Q-learning is tested in two different indoor environments. The environments and corresponding heat maps showing the RSS values for each possible UAV location are generated in a ray tracing software for a more realistic evaluation. The proposed method is compared with the baseline (i.e., location-based) Q-learning algorithm for different UAV speeds in terms of convergence speed, the number of steps taken to reach the victim in the final route, and averaged number of steps per episode.\looseness=-1

\begin{table}[t]
\vspace{-1.5mm}
\renewcommand{\arraystretch}{1.2}
\centering
\caption{Literature Review}
\resizebox{0.49\textwidth}{!}{
\begin{tabular}{p{0.6cm} p{1.3cm} p{1.6cm} p{0.95cm} p{0.55cm} p{1.6cm}} \hline
Ref. & {Input} & Method & Use & Target & Goal \\ \hline
\cite{acuna} & RF signal \& GPS & Random Forest & Outdoor & \cmark & Location estimation \\ \hline 
\cite{Wang2017} & RF signal \& GPS & Deep RL & Outdoor & \cmark & Rescue lost victims \\ \hline
\cite{Iman2016} & GPS \& Sensors & RL & Indoor/ Outdoor & \cmark & Autonomous navigation in an unknown environment \\ \hline
\cite{Pham2018} & GPS \& Sensors & RL & Outdoor & \cmark & Localize immobile victim\\  \hline
\cite{krypto} & RF signal \& GPS & Predefined searching path & Outdoor & \cmark & Rescue victim \\ \hline
\cite{wang} & GPS \& Sensors & Deep RL & Outdoor & \xmark & UAV navigation \\ \hline
\cite{Rodriguez2019} &  GPS \& Sensors & Deep RL & Outdoor & \xmark & Autonomous landing on~a~moving platform \\ \hline 
\cite{KimC15a} & Camera image & ConvNet & Indoor & \cmark & Indoor UAV navigation \\ \hline
\cite{DBLP:journals/corr/SadeghiL16} & 3D CAD models & Image processing \& deep learning & Indoor & \xmark & Collision-free indoor UAV navigation \\ \hline
\cite{Gandhi2017} & Camera image & Image processing \& deep learning & Indoor & \xmark & Collision-free indoor UAV navigation \\ \hline
\cite{gesbert2018} & GPS & RL & Outdoor & \cmark & Maximize sum-rate\\ \hline
\cite{Ferrari2012} & Infrared sensor & RL & Outdoor & \cmark & Detect targets\\ \hline
Our method & RF signal & RL & Indoor/ Outdoor & \cmark & Localize fixed victim\\ \hline
\end{tabular}}
\label{survey}
\end{table}

The remainder of this paper is organized as follows. Section~\ref{sec:Fundamentals} briefly describes the Q-learning algorithm. Simulation setup is introduced in Section~\ref{sec:Setup}. The RSS-based Q-learning algorithm for indoor navigation of UAVs is elaborated in Section~\ref{sec:UAV_navigation}. Experimental results are presented in Section~\ref{sec:Experiments}. Finally, Section~\ref{sec:Conc} concludes this paper.
\section{Background on Q-Learning}
\label{sec:Fundamentals}
As mentioned in Section~\ref{sec:Intro}, RL is a branch of ML that addresses problems where there is no explicit training data available. Q-learning, proposed by Watkins~\cite{watkins}, can be used to learn optimal policies in finite MDPs~\cite{dqn}. This traditional table-based Q-learning maximizes the expected value of the total reward over any and all successive steps by taking action in the current state and follows an optimal policy afterwards. It learns by interacting with the environment 
and approximates a value function of each state-action pair through a number of iterations. The goal is to select the action which has the maximum $Q$-value using the following update rule at~each~iteration:
\begin{equation}
{Q(s,a)\leftarrow(1-\alpha)Q(s,a)+\alpha\big[r(s)+\gamma\max\limits_{a'}Q(s',a')\big]},
\label{eq:q_fn}
\end{equation}
where $s'$ is the state reached from state $s$ after taking action $a$, $\alpha \in (0,1]$ is the learning rate, $r(s)$ is the reward attained for the current state $s$, and $\gamma$ is the discount factor which determines the importance of future rewards. The Q-learning loop is illustrated in Fig.~\ref{RL_algo}. Note that a high $\gamma$ sets priority towards distant future rewards whereas a lower one will force the agent to consider only immediate rewards. After updating the Q-table, the best policy can be obtained by acting greedily in every state by 
\begin{equation}
\label{eq:Update_eqn}
\pi^*=\argmax\limits_\mathbf{a} Q(\mathbf{s},\mathbf{a}).
\end{equation}

\section{Simulation Environment Setup}
\label{sec:Setup}
In this section, we describe the simulation environment for testing the proposed method. We use 
Wireless InSite ray tracing tool, which can provide a deterministic way of characterizing the RSS in indoor scenarios. 
First, we generate two arbitrary floor plans with different complexities and of size $26\times96$~m$^2$ and $76\times58$~m$^2$ (hereinafter referred to as the Scenario 1 and 2, respectively) using the \textit{floorplan} feature of the software. The two floor plans are shown in Fig.~\ref{fig:indoor}. The height is considered to be 3~m for all the walls. We set the UE at an elevation of 1.5~m from the ground.

After generating the floor plans, we run the ray tracing simulations to obtain the RSS at each RX grid with the source being at a specified position. The UE  
\begin{wrapfigure}{r}{5cm}
     \centering \vspace{-8mm}
     \subfloat[Scenario 1.]
         {\includegraphics[width=1\linewidth]{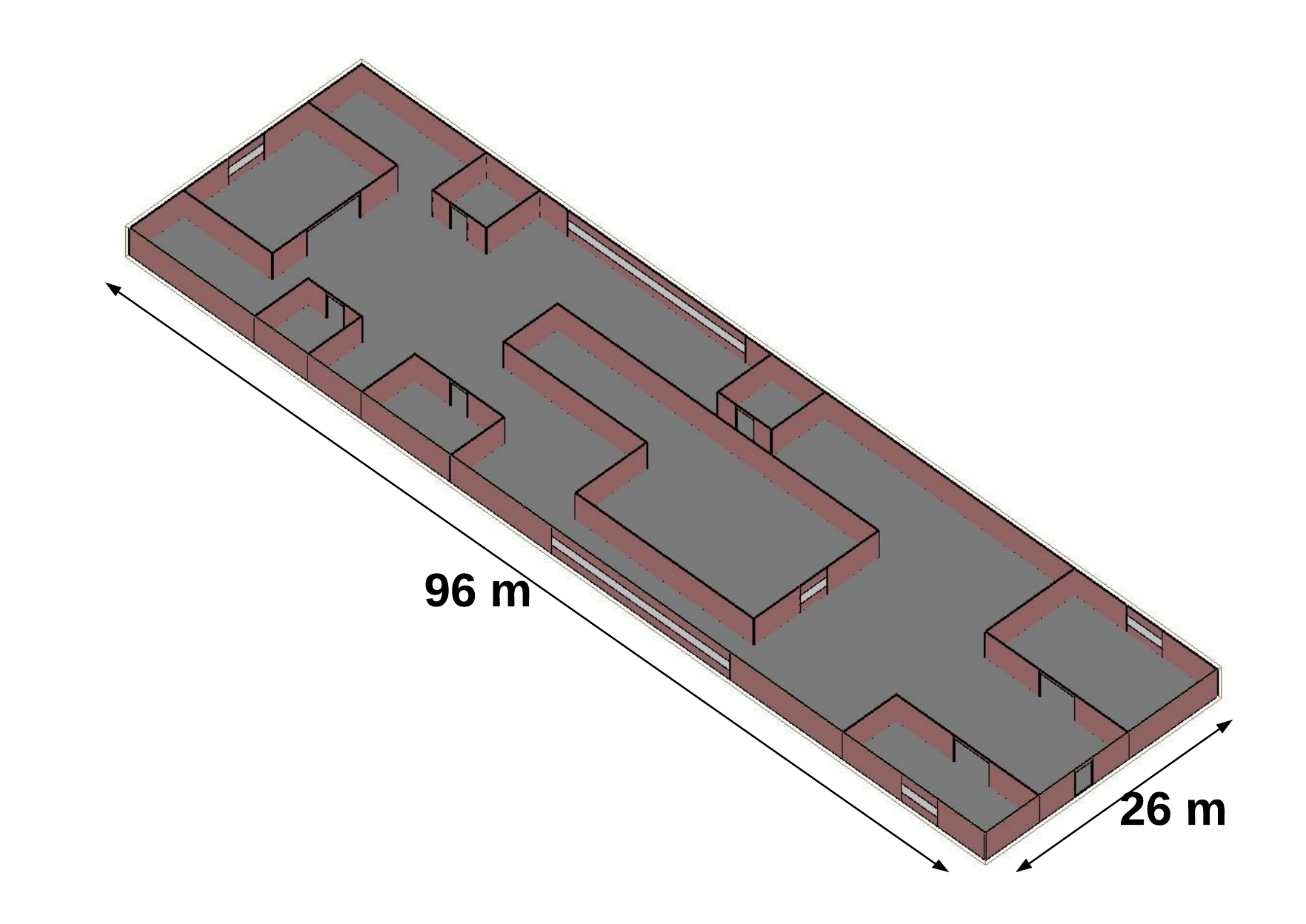}
         }
         \vspace{-2pt}
     \vfill
     \subfloat[Scenario 2.]
         {\includegraphics[width=1\linewidth]{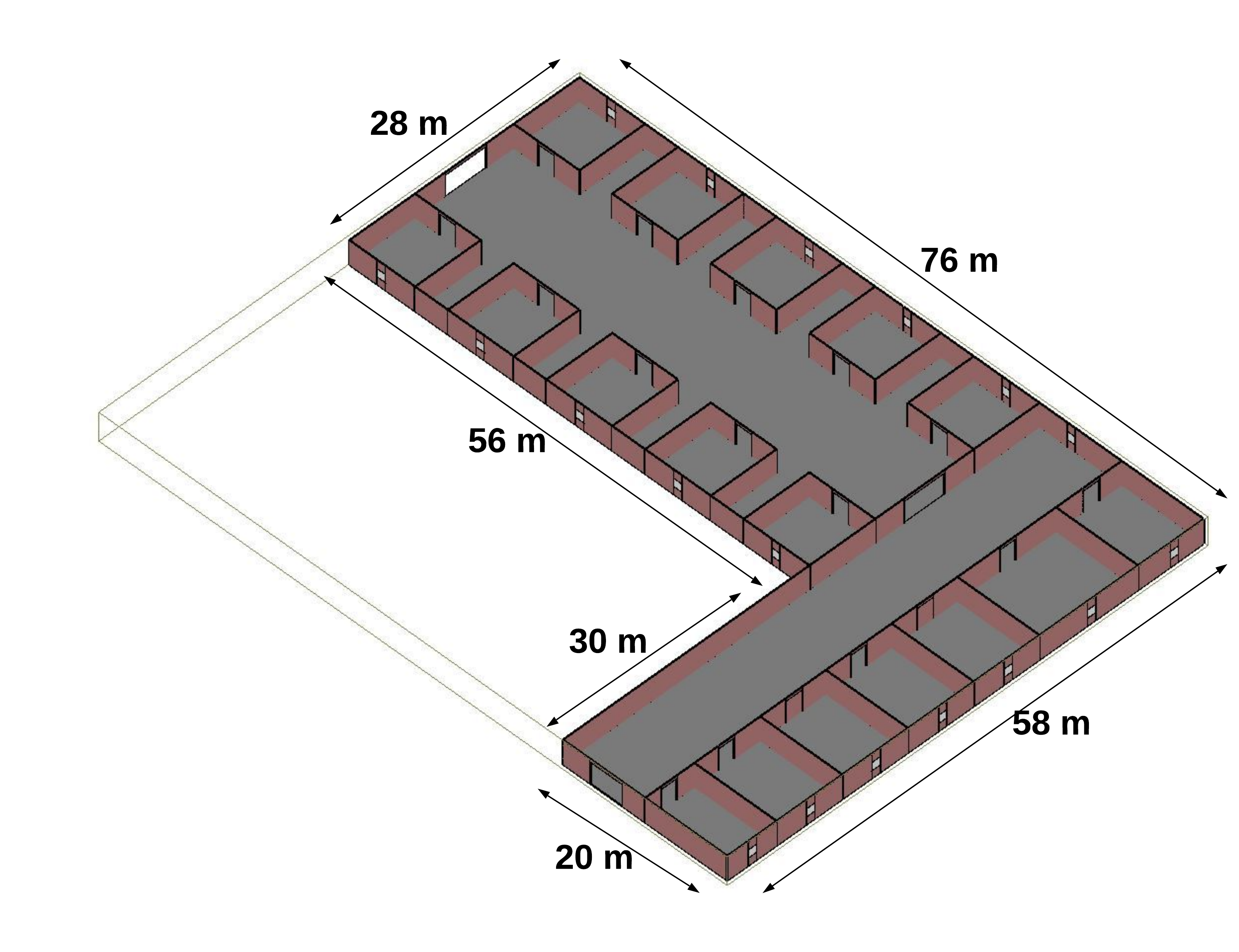}
         }
        \caption{Two indoor environments of size (a) 26~m$\times$96~m and (b) 76~m$\times$58~m. 
        Wall heights are 3 m for both the scenarios.}
        \label{fig:indoor}
        \vspace{-2mm}
\end{wrapfigure}
is assumed to transmit RF signals using 25 dBm transmit power at 2.4 GHz. The RX grids are set $1$~m apart from each other using XY grid option in the software. Half-wave dipole antennas with vertical orientation are used at each RX grid. The maximum antenna gains are considered as 0 dB for both the RXs and UE TX. The height of the grids is considered as 2 m to avoid crashing into obstacles such as tables, cubicles, chairs, etc. The height of the doors is considered to be 2.70 m. Other settings considered in the simulations are as follows. Diffuse scattering mode is disabled.  A maximum of six reflections and one diffraction were allowed. \looseness=-1

We also create the same floor plans in MATLAB.  Then, we transfer the resulting RSS maps from the ray tracing to MATLAB for use in the navigation simulations. For simplicity, we assume that the UAV flies at a constant altitude. Thus, all the allowable actions that are separated by $45^{\circ}$ lie in the $xy$ plane. We consider three different UAV speeds, namely, 1 m/s, 2 m/s and 4 m/s. Most commercial drones available in the market come with a maximum speed limit of 40 mph or 18 m/s \cite{dronespeed}. So, our assumptions about the UAV speeds are reasonable.  
Note that, for simulation purposes, floor areas are partitioned into the grids and hence the UAV is forced to move from the center of one grid to that of the other. That is, if the UAV speed is set to $v$ m/s and the UAV makes a diagonal movement, e.g., moves from the grid index (1,1) to (2,2), its speed will be $v\sqrt{2}$ m/s. For simplicity, while presenting the results in Section~\ref{sec:Experiments}, we will refer to the UAV speed as $v$ m/s independent of the movement direction. We also assume that the UAV senses the RSS intermittently with a 1-second interval. In other words, the UAV will detect the RSS only when it reaches a new location. Such a sensing method will help the UAV save battery power.
\section{UAV Navigation Using RSS-Based Q-Learning}
\label{sec:UAV_navigation}
In this section, we introduce the RSS-based Q-learning method for the navigation of a UAV to a wireless source. In the location-based Q-learning algorithm, states and rewards are defined based on the location of the agent, i.e., GPS coordinates. This method is not suitable for use indoors where the GPS signal is not available. It also requires the exact coordinates (or an accurate mathematical representation of the position) of the target which is also not available in most of the SAR scenarios. On the other hand, in our proposed approach, states are defined based on the RSS values at each particular grid or UAV location. RSS values are also used in the definition of rewards allowing the navigation of the UAV towards the target by providing a reasonable representation of the target location. 
\subsection{State and Reward Definitions}
The UAV starts from an initial position and detects the RSS at that position. A state label is assigned to this particular RSS value. Based on the fact that no two grids (separated by 1 m in this case) will have the same RSS value, each location is represented uniquely by a state. Then, the UAV takes an action depending on the strategy of the algorithm in use and moves to a new location. The reward is defined as the difference between the RSS values associated with the latest and the previous position, i.e., $\textrm{RSS}_t-\textrm{RSS}_{t-1}$, so that higher rewards are obtained when there is an increase in the RSS. Next, a state label is assigned to the new location based on the new RSS value, and the Q-table is updated using the update equation in~\eqref{eq:q_fn}. \looseness=-1 

It is worth noting that there may be small deviations from the previous RSS values at the next visits to the same grid. These deviations may be due to the imprecise steps taken by the UAV or some small changes in the environment or the source position. Since the states are defined based on the RSS values, this situation may lead to representing a single grid by multiple states, which, in turn, delays the convergence of the algorithm. As a solution to this problem, states can be defined as the neighborhood of the detected RSS values. If the RSS value of a new location does not lie in an already defined interval, then a new state is defined; otherwise, the same state (as one of the previous states) is attained. That is, if a state is labeled as $s_i$ for the RSS value detected at time $t$, then the same state will be attained whenever a new RSS value is detected within the range $(\textrm{RSS}_t-Th,\textrm{RSS}_t+Th)$. The threshold $Th$ should be defined in such a way that the state will remain unchanged provided that the UAV hovers inside the boundaries of a grid.

Alternatively, states can be defined based on a set of RSS intervals determined before running the algorithm. A sufficiently wide range of RSS values can be divided into a number of discrete segments, and the states are assigned based on which segment the RSS at a particular location falls into. This technique may result in a small number of states, but it creates another interesting problem. For instance, two or more different locations in the indoor environments can be of the same state due to having close RSS values. Hence, a \textit{good} action at one location can be a \textit{bad} action at another one leading the UAV to crash. Consequently, instability may be observed in the Q-table update process. 
For simplicity, we assume a static environment and use the special case of the above-mentioned solution with $Th=0$, i.e., each RSS value detected at a location is given a single state label. 

Each episode ends when the UAV is close enough to the target. We assume an episode ends when the distance between the UAV and the victim is less than 2 m. Using free-space path loss model \cite{fspl}, we calculate this RSS threshold to be -21 dBm. Note that, if the distance between the UAV and victim is less than 2 m and there is a wall between them, the RSS value pertinent to that position will be far less than -21 dBm due to the presence of the wall.



Collisions are major problems for autonomous UAV navigation, and can be avoided using a range sensor or video camera-based systems as suggested in~\cite{KimC15a,Gandhi2017}. We do not address this problem in this study. However, to simulate the possible solutions, each time before the agent takes a new action, we check if that action leads to a crash. If so, the action is dropped from the list of possible actions, and another action is picked. The overall Q-learning process is summarized in Algorithm 1.

\subsection{$\epsilon$-greedy Method}
To overcome the exploration-exploitation dilemma in Q-learning, we deploy $\epsilon$-greedy method. 
The main idea of $\epsilon$-greedy method is to choose a random number from [0,1] and check whether it is greater than $\epsilon$. If it is lower than $\epsilon$, the agent takes random action; otherwise, it goes with the “greedy” action that has the highest $Q$-value. It is shown in \cite{dqn}, that 
\begin{wraptable}{r}{4.6cm}
\centering \small 
\renewcommand{\arraystretch}{1.2}
\vspace{-5mm}
\caption {Simulation parameters.}
\label{Tab:Sim_par}
{\begin{tabular}{lc}
\hline
Parameter & Value \\
\hline
$\epsilon_{\rm max}$, $\epsilon_{\rm min}$ & 1, 0.01  \\
$\alpha_{\rm max}$, $\alpha_{\rm min}$ & 0.5, 0.05  \\
$\eta$ & $10^{-5}$ \\
UE Transmit power & 25 dBm \\ 
$\gamma$ &  0.98\\
\hline
\end{tabular}}
\vspace{-2.5mm}
\end{wraptable}
starting with a high $\epsilon$ and then decreasing it with episodes can provide better convergence performance. Hence, we also start with $\epsilon=1$ and decrease it exponentially with a decay factor $\eta$ with iteration number. To increase the importance of the future rewards, we set discount factor $\gamma$ to be 0.98. The learning model parameters used in this study are specified in Table~\ref{Tab:Sim_par}. In each iteration, the agent or UAV in our case, starts from an initial location and traverses through the indoor scenario.
If the UAV detects the UE, it will get a reward of 1000. Once the UAV finds the target, the current episode finishes and the new one starts. Since the UAV becomes more experienced as it moves through the indoor environments, we also decay $\alpha$ exponentially with $\eta$. 

\begin{algorithm}[t]\small 
	\caption{RSS-based Q-learning for indoor UAV navigation.}
    \label{alg:Alg2}
	\begin{algorithmic}[1]
		\STATE  start from an initial location and obtain associated\\ state of that particular location by sensing the RSS
		\STATE \textbf{repeat} (for each step):
		\STATE \hspace{0.3cm} \textbf{if} $\epsilon$ $\geq$ $\epsilon_{\rm min}$
		\STATE \hspace{0.6cm} $\epsilon$=$\epsilon$ $\times$ $\exp{(-\eta)}$ \textbf{end if}\\
		\STATE \hspace{0.3cm} \textbf{if} $\alpha$ $\geq$ $\alpha_{\rm min}$
		\STATE \hspace{0.6cm} $\alpha$=$\alpha$ $\times$ $\exp{(-\eta)}$ \textbf{end if}\\
		\STATE \hspace{0.3cm} choose $a$ using $\epsilon$-greedy policy
		\STATE \hspace{0.3cm} take action $a$, observe $s'$
		\STATE \hspace{0.3cm} check $s'$ for possible obstacle(s)
		\STATE \hspace{0.3cm} \textbf{while} any obstacle at $s'$ \textbf{do}
		\STATE \hspace{0.6cm} leave $a$ and select any other action randomly, \textbf{end while}
		\STATE \hspace{0.3cm} Calculate reward for taking action $a$ by  subtracting RSS\\\hspace{0.3cm}  associated with state $s$ from state $s'$
		\STATE \hspace{0.3cm} update $Q$-value using \eqref{eq:q_fn}
        
		\STATE \hspace{0.3cm} $s \leftarrow s'$
		\STATE \textbf{until} $s$ is terminal
	\end{algorithmic}
\end{algorithm}

\begin{prop}
RSS-based Q-learning algorithm is an MDP.
\end{prop}
\begin{proof}
According to \cite{dqn} and \cite{doubleq}, an MDP has five components: 1) finite states, 2) a finite set of actions, 3) a transition probability, 4) an immediate reward function, and 5) a decision epoch set that can be either finite or infinite. In our proposed algorithm, if the indoor scenario is of finite area, the total number of unique states will also be finite. The total number of allowable actions is eight and the UAV can choose an action by $\epsilon$-greedy method. The reward function is defined as the difference between the RSS value of the current state and previous state and finally, the UAV takes decisions until it finds the victim, which leads to a finite decision epoch. Thus, we can conclude that the proposed indoor navigation framework is an MDP.\looseness=-1 
\end{proof}
\begin{cor}
\textit{The Q-learning algorithm in the proposed RSS-based indoor navigation system will converge to an optimal action-value function with probability one.}
\end{cor}

\begin{proof}
In our proposed method, the states and actions are finite and we consider $\gamma$ to be less than one. The reward function is finite and $\alpha \in [0,1]$. All the $Q$-values are updated and stored in tables. Q-tables of both RSS-based and location-based algorithms get an infinite number of updates. Thus we fulfill all the conditions mentioned in~\cite{watkins} for convergence.
\end{proof} 
\subsection{Limitations}
There are a few limitations in our simulation setup which we plan to address in our future research. We assume a stationary indoor environment where the victim is stagnant, which might not always be the case. In fact, in case of emergencies, the victims might switch their locations abruptly and randomly for safety purposes. In addition, frequent sharp turns while traversing will cost the UAV with more battery power. We overlook this non-trivial issue intentionally for the sake of simplicity. We will consider battery constraints and a dynamic environment for the navigation of UAVs in our future research.
\section{Experiments and Results}
\label{sec:Experiments}

 We first investigate the trajectories followed by the UAV in both scenarios using the RSS-based algorithm. We consider a location-based Q-learning algorithm as the baseline, where we assume that the UAV can track its indoor location and the location of the target is known beforehand. Apart from these, the reward is defined as $(1/D_t)$ in the location-based algorithm, where $D_t$ is the Euclidean distance between the UAV and the victim after taking an action at time $t$. In this way, the UAV will try to minimize its distance from the victim through the iterations. Note that, for UAV speeds greater than 1 m/s, the UAV may not land on the exact location of the victim for different starting points. 
 Hence, we consider that an episode ends when $D_t$ is less than 2 m as in the case of the RSS-based Q-learning. The main differences between the two methods are summarized in Table~\ref{Tab:rss_vs_location}.
 
 \begin{figure}[t]%
 \vspace{-3mm}
    \centering
    \subfloat[]{{\includegraphics[width=0.5\linewidth]{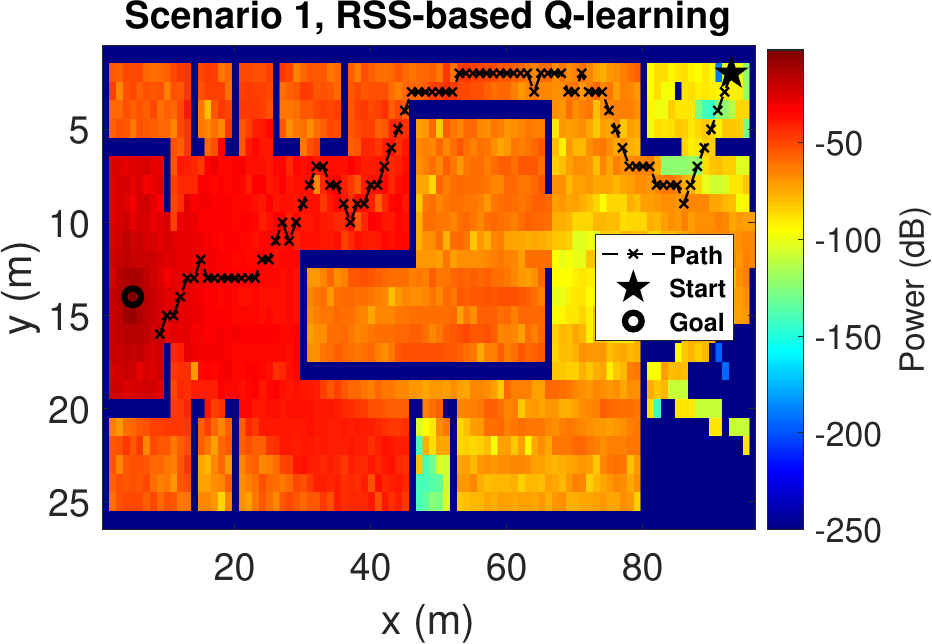} }}%
    \subfloat[]{{\includegraphics[width=0.5\linewidth]{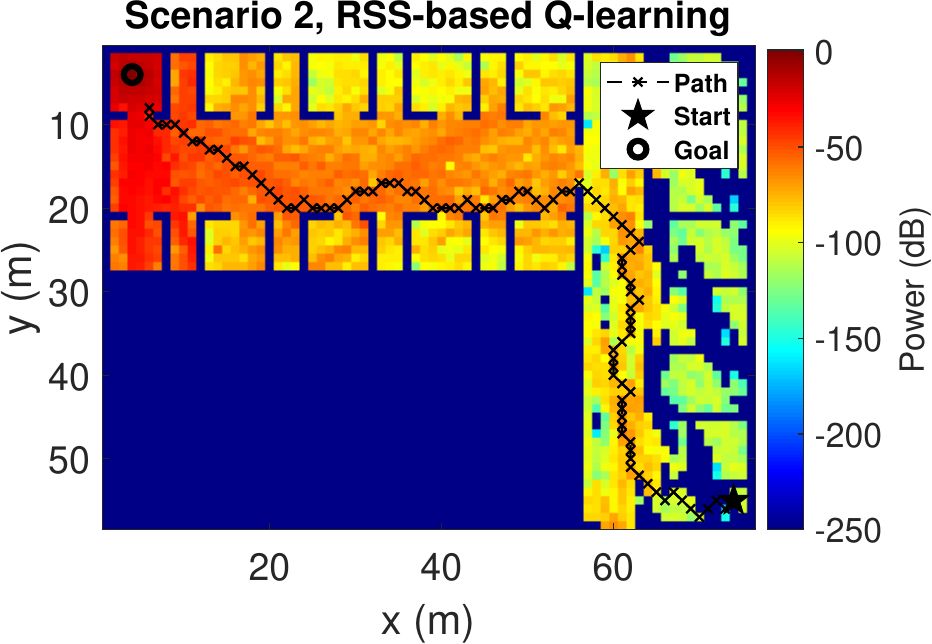} }}%
    \qquad
    
    \subfloat[]{{\includegraphics[width=0.49\linewidth]{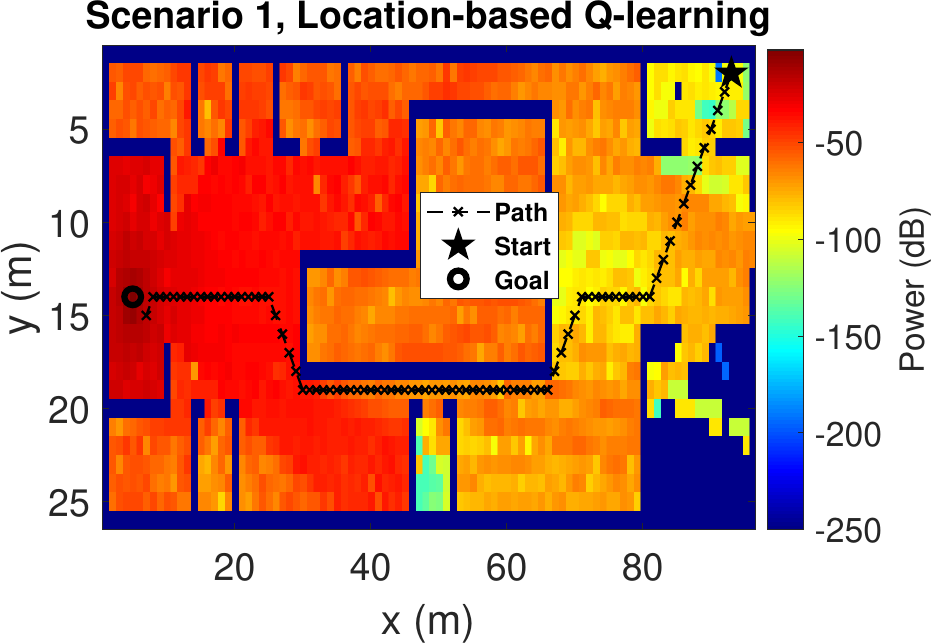} }}%
    \subfloat[]{{\includegraphics[width=0.49\linewidth]{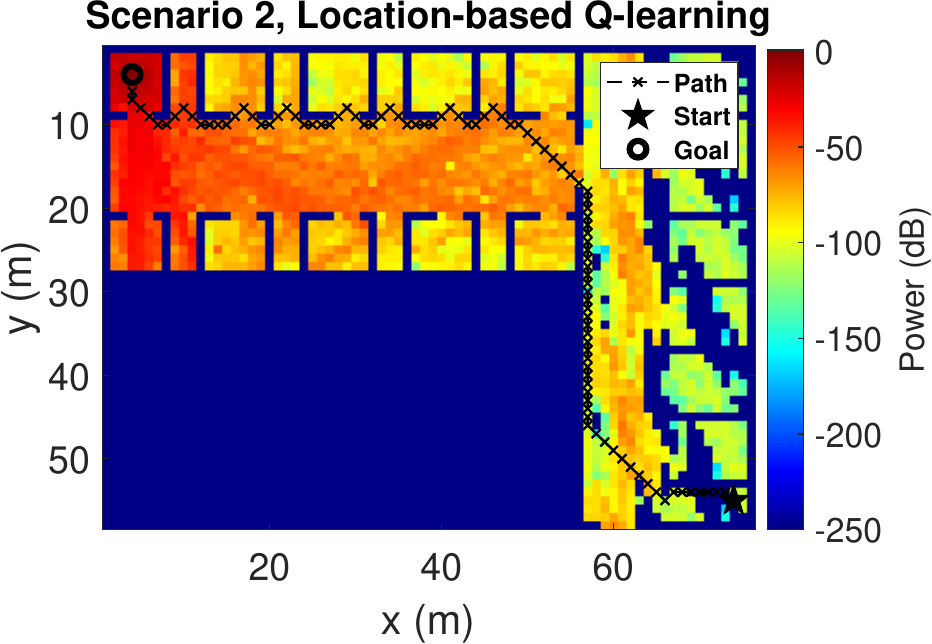} }}%
    \caption{UAV trajectories for (a)-(b) RSS-based Q-learning and (c)-(d) location-based Q-learning for the two indoor scenarios. UAV speed is 1 m/s.} 
    \label{fig:traject_indoor}%
\end{figure}
 

\begin{table}[]
 \vspace{-3mm}
\centering
\renewcommand{\arraystretch}{1.3}
\caption {Comparison between the RSS-Based and Location-Based Q-Learning Algorithms.}
\label{Tab:rss_vs_location}
\resizebox{0.49\textwidth}{!}{
\begin{tabular}{cccc}
\hline
Technique & \begin{tabular}[c]{@{}c@{}}State \\ definition\end{tabular} & \begin{tabular}[c]{@{}c@{}}Reward \\ definition\end{tabular} & \begin{tabular}[c]{@{}c@{}}Convergence\\ criteria\end{tabular} \\ \hline
RSS-based & $\textrm{RSS}_t$ & $\textrm{RSS}_t-\textrm{RSS}_{t-1}$ & $\textrm{RSS}_t>-21$ dBm \\ 
Location-based & (x,y) coordinate & $1/D_t$ & $D_t<2$ m \\ \hline
\end{tabular}}
\end{table}

 \begin{figure}[t!]%
\vspace{-4mm}
    \centering
    \subfloat[]{{\includegraphics[width=0.46\linewidth]{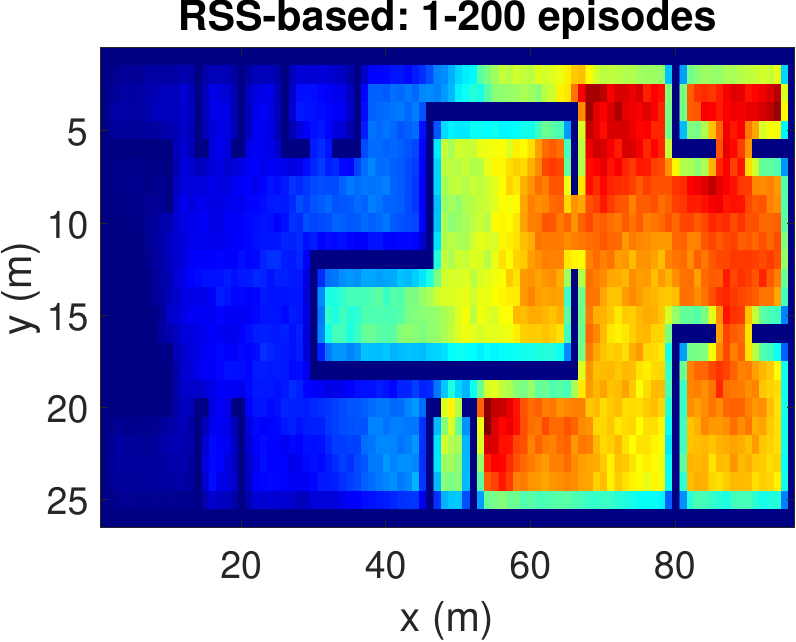}}}%
    \hspace{1mm}
    \subfloat[]{{\includegraphics[width=0.46\linewidth]{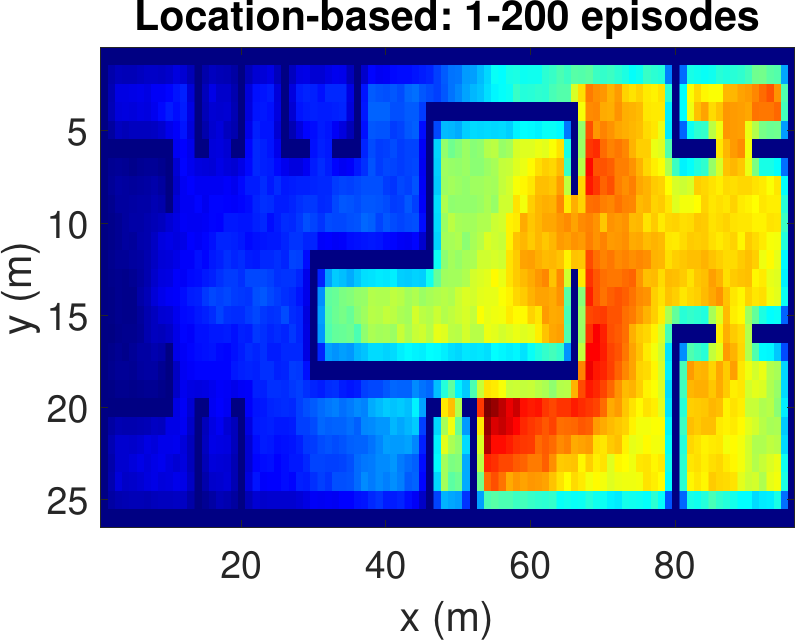}}}%
    \qquad
    \subfloat[]{{\includegraphics[width=0.46\linewidth]{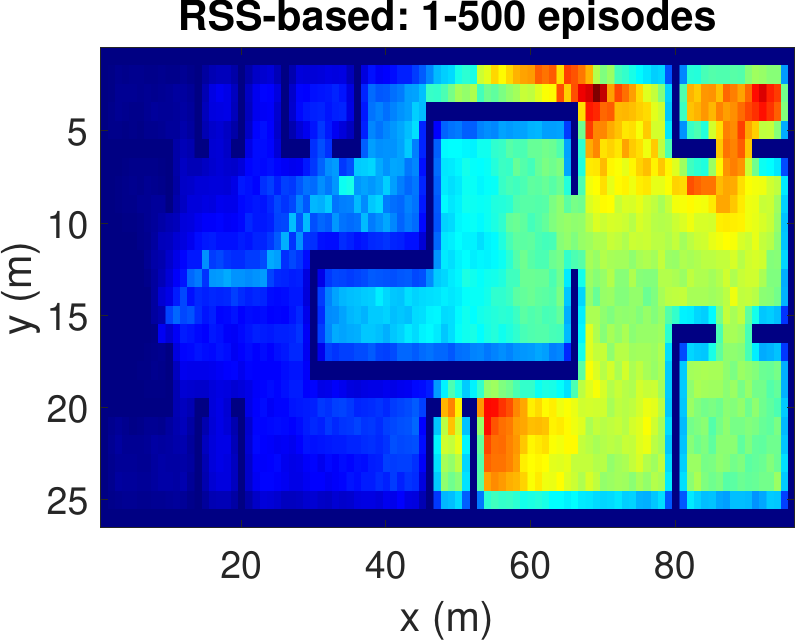}}}%
    \hspace{1mm}
    \subfloat[]{{\includegraphics[width=0.46\linewidth]{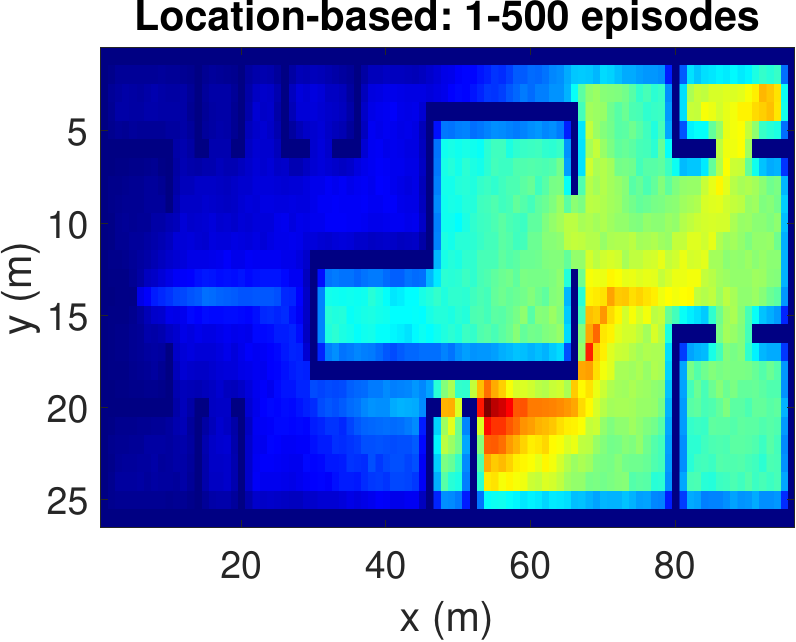}}}%
    \qquad
    \subfloat[]{{\includegraphics[width=0.46\linewidth]{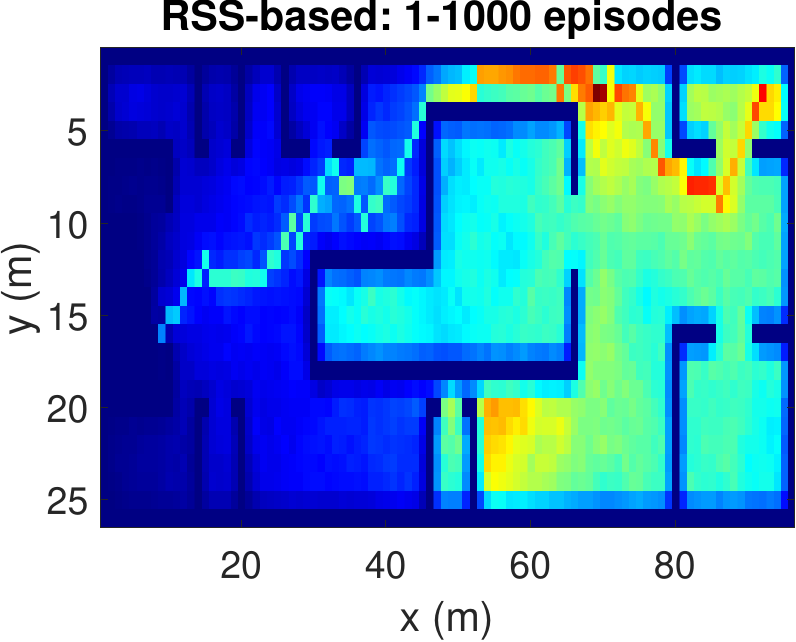}}}%
    \hspace{1mm}
    \subfloat[]{{\includegraphics[width=0.46\linewidth]{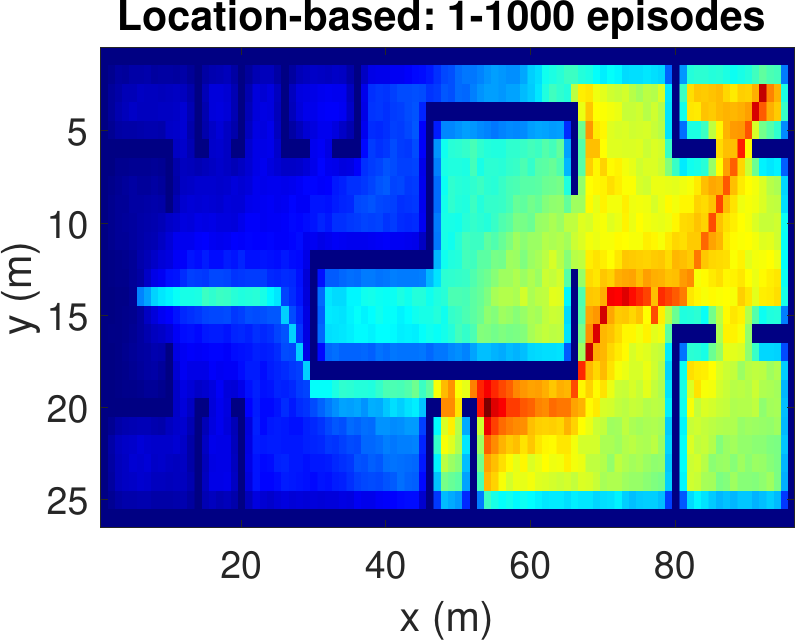}}}%
    \caption{The heat map of the number of visits of location-based and RSS-based Q-learning algorithms in Scenario 1 for different episode intervals. UAV speed is 1 m/s. Brighter colors indicate more visits.}
    \label{fig:visit_heatmap}%
\end{figure}

\begin{figure}[t]%
    \centering
    \subfloat[]{{\includegraphics[width=0.48\linewidth]{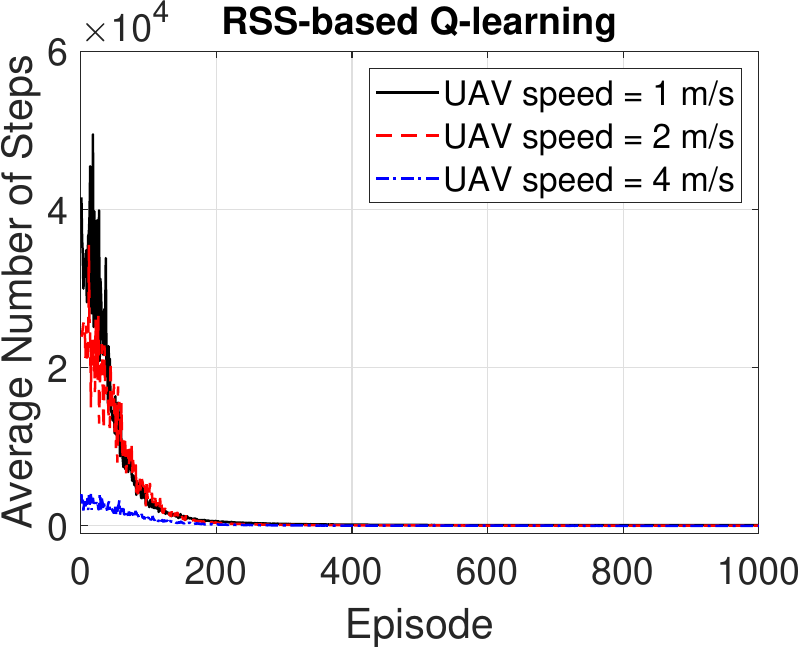}}}%
    \hspace{0.2cm}
        \subfloat[]{{\includegraphics[width=0.48\linewidth]{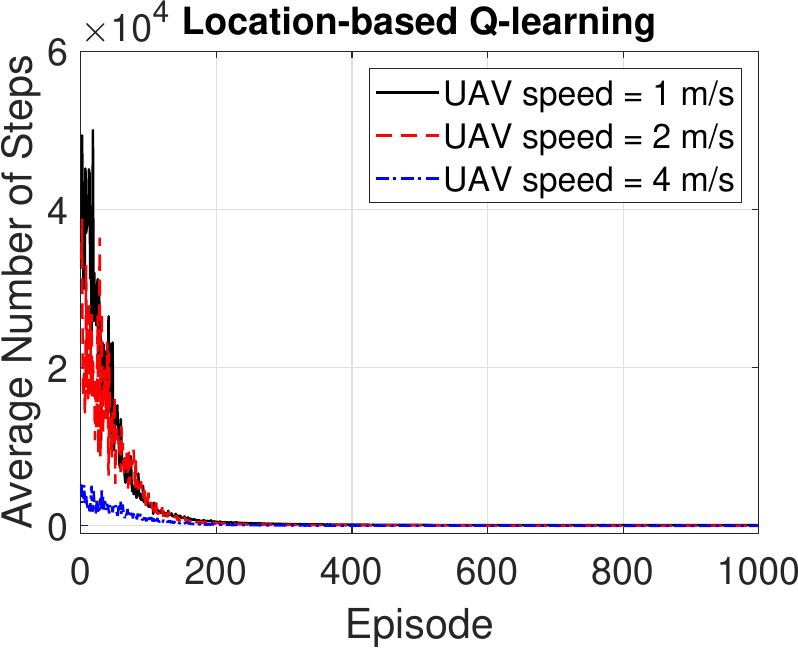}}}%
    \caption{Average number of steps per episode in (a) RSS-based and (b) in location-based Q-learning algorithms for different UAV speeds in  Scenario 1.} 
    \label{steps_episode_q_vs_dq}%
\end{figure}

 The resulting UAV trajectories for the RSS-based algorithm are shown in Fig.~\ref{fig:traject_indoor}(a) and Fig.\ref{fig:traject_indoor}(b). The UAV speed is considered to be 1 m/s. The UAV starts from the initial location (93 m, 2 m) in 
 Scenario~1, and from the location (74 m, 55 m) in Scenario~2. The victim is considered to be situated at (5 m, 14 m) in 
 Scenario~1, and at the location (4 m, 4 m) in Scenario~2. In both scenarios, we observe that the trajectories tend to avoid the regions with low RSS values. Since the reward is defined as the difference between the RSS values at successive states, the UAV shows an inclination to have higher RSS values at the next steps rather than finding the victim with the smallest path. We see the same trend for other simulations with different starting positions. Sensing the paths with higher RSS values eventually leads the UAV towards the victim. Although the UAV does not know the victim's location, it can successfully reach the destination.\looseness=-1
 
The trajectories associated with the location-based Q-learning are shown in Figs.~\ref{fig:traject_indoor}(c) and \ref{fig:traject_indoor}(d) for Scenario 1 and 2. The UAV starting points and target locations are kept the same as those of the RSS-based Q-learning experiments. We observe that the UAV tries to find the shortest path towards the victim in both scenarios as expected. Note that, in Fig.~\ref{fig:traject_indoor}(d), the UAV tends to enter some of the compartments. This is due to the fact that the points inside the compartments are nearer to the victim from any other point in the hallway area. For higher speeds, UAV avoids those points since the overall distance covered by the UAV will be increased otherwise. 
 

To have a better understanding of the learning processes, we investigate the relative frequency of the state visits through the episodes. Heat maps (averaged over 100 runs) in Fig.~\ref{fig:visit_heatmap} show the results for three different episode intervals. Comparing Fig.~\ref{fig:visit_heatmap}(a) and Fig.~\ref{fig:visit_heatmap}(b), we observe that the location-based Q-learning visits nearby locations to the starting point more frequently in the first 200 episodes than its RSS-based counterpart. This is because the RSS-based method tries to find the locations that provide higher signal strength and RSS values at different locations are unique. As a consequence, RSS-based method learns better policies faster. On the other hand, location-based Q-learning focuses on finding the shortest route and two or more locations might have same distances from the target. Hence, location-based method needs more explorations. From Fig.~\ref{fig:visit_heatmap}(c) and Fig.\ref{fig:visit_heatmap}(d), which show the frequency of the state visits in the first 500 episodes, we can also conclude that the RSS-based method finds the optimal policy earlier than the location-based method. 
Lastly, as it is clear from Figs.~\ref{fig:visit_heatmap}(e) and~\ref{fig:visit_heatmap}(f), both methods learn optimal policies during the the first 1000 episodes.

Fig.~\ref{steps_episode_q_vs_dq} shows the average number of steps taken per episode by the UAV to reach its goal for different speeds in Scenario 1. The number of steps required in each episode is averaged over 100 realizations. As expected, the number of steps decreases with the episode index. The UAV learns the representation of the indoor environment better as it becomes more experienced and hence, it requires fewer steps to reach the goal. The UAV can move to fewer states as its speed increases, and thus, the Q-learning algorithms tend to converge  quicker with higher UAV speeds. 
Moreover, we observe that the RSS-based navigation converges within about the same number of episodes as the location-based method. Similar to the observations in Fig.~\ref{fig:visit_heatmap}, since the RSS-based technique only focuses on getting higher RSS values as rewards, it quickly learns to skip the states that provide lower RSS values. Meanwhile, the location-based Q-learning treats every possible state equally and hence ends up with getting higher average steps during the early episodes.

Next, we explore the convergence time of the algorithms for different UAV speeds. 
We record the trajectory followed by the UAV to reach its goal for each episode. If the UAV follows the same path for three consecutive episodes, we conclude that the Q-table is converged. The time elapsed until the convergence of the Q-tables is averaged over 100 executions. The results are shown in Fig.~\ref{complex_time}. 
Similarly to the above results, since the number of allowable actions decreases with the UAV speed, convergence time decreases for both algorithms. We observe that the RSS-based algorithm shows competitive performance in terms of convergence time with the location-based algorithm, especially for higher UAV speeds. 

 \begin{figure}[t]
     \centering
         {\includegraphics[width=0.79\linewidth] {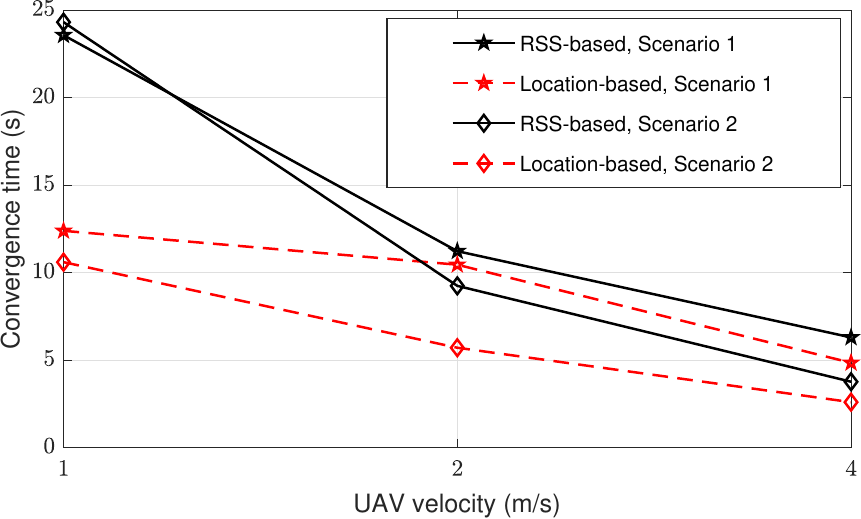}
         }
        \vspace{-6pt}
        \caption{Convergence time of the RSS-based and location-based Q-learning algorithms for different UAV speeds.} 
        \label{complex_time}
        \vspace{-4pt}
\end{figure}

Finally, we provide the total length of the final trajectories in  Table~\ref{Tab:LengthSteps}.
Since Scenario 2 consists of longer hallways and include compartments, the UAV needs to take more steps to reach the goal when compared to Scenario 1. Overall, our proposed technique provides very close results to the location-based algorithm in terms of the number of steps in the final trajectory. However, having even the same number of steps does not always imply having the same computational time or path length. This is due to the fact that diagonal movements take more time than the movements in left-right and up-down paths. Since the location-based algorithm results in more straight trajectories as shown in Fig.~\ref{fig:traject_indoor}, the total final path length and flight time will be smaller than those of the RSS-based algorithm. \looseness=-1

\renewcommand{\arraystretch}{1.2}


\begin{table}[]
\centering
\renewcommand{\arraystretch}{1.3}
\caption {The Total Length of the Final Trajectories of the Q-Learning Algorithms for Different UAV Speeds and Scenarios.}
\label{Tab:LengthSteps}
\resizebox{0.49\textwidth}{!}{
\begin{tabular}{ccccc}
\hline
 & \multicolumn{4}{c}{Total path length (m)} \\ \cline{2-5} 
 UAV speed& \multicolumn{2}{c}{Scenario 1} & \multicolumn{2}{c}{Scenario 2} \\ \cline{2-5} 
 (m/s) & \multicolumn{1}{l}{RSS-based} & \multicolumn{1}{l}{Location-based} & \multicolumn{1}{l}{RSS-based} & \multicolumn{1}{l}{Location-based} \\ \hline
1 & 105.88 & 94.70 & 122.61 & 119.71 \\ 
2 & 104.22 & 101.25 & 121.65 & 121.05 \\ 
4 & 102.91 & 98.63 & 126.22 & 114.63 \\ \hline
\end{tabular}}
\end{table}

\textcolor{red}{}

\vspace{-4mm}

\section{Conclusion}
\label{sec:Conc}
In this paper, we studied the problem of detecting or rescuing a victim in a GPS-denied indoor environment using the 
RSS of the RF signals sent by the victim's smart devices. We envisioned a rescue system by deploying a UAV, which will navigate through the indoor environments using Q-learning techniques. We presented simulation results for two indoor scenarios with different complexities. We also compared our proposed technique with the location-based Q-learning and find that RSS-based Q-learning provides competitive performance without requiring the UAV and target location information. Our results show that the RSS-based Q-learning shows less fluctuations during training than the location-based method. The convergence time decreases with the increasing UAV speed for both methods, and the RSS-based technique learns the environment earlier than its location-based counterpart.\looseness=-1




\vspace{-2mm}
\bibliographystyle{IEEEtran}
\balance
\bibliography{IEEEabrv,references}

\end{document}